\documentclass{l4dc2024}

\title[MPC \textit{vs}. NNs]{Mapping back and forth between model predictive control and neural networks}
\usepackage{times}
\usepackage{subcaption}
\usepackage{caption}
\usepackage{subcaption}
\usepackage{wrapfig}
\usepackage{epstopdf}

\newcommand{\ts}[1]{{\textnormal{#1}}}
\newcommand{\ie}{\emph{i.e.,}\ }
\newcommand{\eg}{\emph{e.g.}\ }

\newcommand{\Rset}{\mathbb{R}}

\newcommand{\mc}{\mathcal}
\newcommand{\mb}{\mathbf}

\usepackage{enumerate}
\newcounter{subeqn} %

\author{%
 \Name{Ross Drummond} \Email{ross.drummond@sheffield.ac.uk}\\
 \Name{Pablo R. Baldivieso-Monasterios} \Email{p.baldivieso@sheffield.ac.uk}\\
 \addr Dept. of Automatic Control and Systems Engineering,  University of Sheffield,  Sheffield, S1 3JD, UK. \\
  \AND
 \Name{Giorgio Valmorbida} \Email{giorgio.valmorbida@centralesupelec.fr}\\
 \addr L2S, CentraleSup\'elec, CNRS, Universit\'e Paris-Saclay, Inria-Saclay Projet DISCO, Gif--sur--Yvette
91192, France.%
}
\begin{document}

\maketitle

\begin{abstract}%
Model predictive control (MPC) for linear systems with quadratic costs and linear constraints is shown to admit an exact representation as an implicit neural network.  A method to ``unravel'' the implicit neural network of MPC into an explicit one  is also introduced. As well as building links between model-based and data-driven control, these results emphasize the capability of implicit neural networks for representing  solutions of optimisation problems, as such problems are themselves implicitly  defined functions. 
\end{abstract}


\section*{Introduction}
  For control problems beyond the capabilities of the classical methods, such as PID controllers, two  important benchmarks are model predictive control (MPC) and reinforcement learning policies built upon neural networks (NNs). These two approaches are generally quite distinct; MPC is a \textit{model-based} method which uses a physical model of the system to predict how the trajectory of the system will evolve in the future and then solves an optimisation problem to optimally control the system whilst ensuring some control and safety constraints. 
  Neural networks are the foundations of many \textit{data-driven} control methods, with the neural networks trained to control the system by optimising over trajectory data  from experiments. 
  The properties of these two methods are also distinct; MPC can provide strong guarantees on robustness and optimality whilst NNs can be applied to problems where the problem's objectives, constraints and models may be unknown.

  In spite of their differences, there has been strong interest in linking MPC and NN controllers. Typically, the motivation to link the two has been driven by the need to improve the scalability of MPC, since evaluating a NN is typically much simpler than computing the solution of an MPC problem.  Various schemes have been proposed to link MPC and NNs in this way, including \cite{fahandezh2020proximity,xu2021lattice,karg2020efficient,ferlez2020aren,aakesson2006neural,kittisupakorn2009neural,xiao2022learning,hertneck2018learning,soloperto2020augmenting} and \cite{ahn2022model} which used imitation learning. An overriding theme of these approaches is to exploit the fact that the optimal solution of linear quadratic MPC is a piece-wise linear function, and so a piece-wise linear function in the form of a neural network should approximate it well. As these studies show, this motivation is well-founded. 
   However, even then, approximating MPC with a NN generally leads to  approximation errors, and it is only recently that such errors have been bounded
 \cite{drummond2022bounding, fabiani2022reliably, schwan2023stability}. In some special cases, exact representations of  MPC policies using explicit  NNs have been obtained; however strong restrictions on the system dynamics have been imposed in these cases \cite{moritz3}, for example in \cite{moritz2} where the analysis was restricted to single state systems .

  Together, these results highlight the growing interest in linking the disparate control policies based upon NNs and MPC. Not only will this help improve the scalability of MPC but it also promises to strengthen the bridges between data-driven and model-based control. 
  
\textbf{Contributions:} The main results of this paper are:
\begin{enumerate}
    \item An analytic representation for an MPC law for linear systems with a quadratic cost and linear input/state constraints in terms of implicit neural networks.
    \item A method to unravel the implicit neural network for MPC into an explicit NN with a bounded approximation error.
    \item A procedure to reconstruct the MPC cost from an explicit NN.
\end{enumerate}
 The main insight of the paper is that linear quadratic MPC (as in a control policy where the control action is determined by solving a quadratic programming problem such as Definition \ref{def:mpc-qp}) has an exact representation as an implicit NN (with the implicit NN structure detailed in \cite{el2021implicit}- although also referred to as equilibrium networks \cite{bai2019deep,revay2023recurrent}).  This is achieved by building upon the representation in terms of linear complementarity problems detailed in  \cite{valmorbida2023quadratic}.  Figure \ref{fig:schematic} illustrates this connection between MPC and NNs mapped out in this paper.  The focus on implicit neural networks explored in this paper follows from the observation that most control policies defined by the solutions of optimisation problems, such as MPC, are themselves implicitly defined functions, and so implicit neural networks should  approximate them well. 
 
The presented results highlight the benefits of focusing on \textit{implicit-to-implicit} mappings for linking MPC and NNs, rather than the \textit{explicit-to-explicit} mappings  considered previously, e.g. in \cite{fahandezh2020proximity,xu2021lattice,karg2020efficient,ferlez2020aren,aakesson2006neural,kittisupakorn2009neural,hertneck2018learning,soloperto2020augmenting}.  The main result of this paper is similar in spirit to \cite{moritz2} on generating exact NN representations of MPC, but by focusing on the implicit mapping, it does not impose strict restrictions on the system, such as the dimension of the state-space being one. With our second result, the obtained implicit MPC is ``\textit{unravelled}'' into an  approximating explicit NN, which can be simply evaluated and has the classical feed-forward NN architecture. Together, these results  point towards a framework to help unify model based control and data-driven control methods built on top of neural networks. As well as helping to improve the computational scalability of MPC, bringing model-based and data-driven methods closer together also promises to improve the data-driven methods' explainability (by associating the black-box of the NN with a predictive model) and robustness (by connecting their stability analysis for control purposes to the well developed field of robust control).

\begin{figure}
         \centering
         \includegraphics[width=0.9\textwidth]{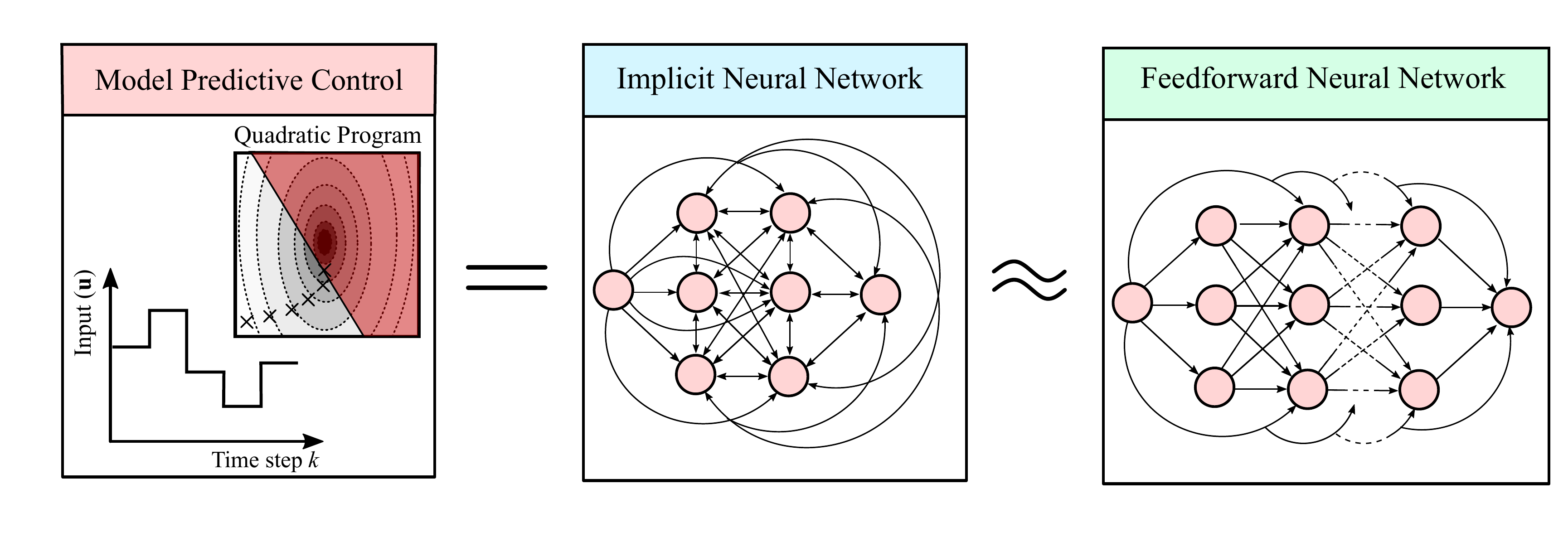}
         \caption{Linear quadratic model predictive control (MPC) is shown to be structured as an implicit neural network. This implicit neural network is then approximated by a feed-forward one. This neural network structure builds upon the results of \cite{valmorbida2023quadratic} on the linear complementarity formulation of the \textit{MPC-QP} problem.
         }
         \label{fig:schematic}
     \end{figure} 

 \subsubsection*{Notation}
 Real matrices $A$ of size $m \times n$ are denoted $A \in \mathbb{R}^{n \times m}$. Positive-definite matrices $A \succ 0$ of size $n$ are denoted $A \in \mathbb{S}^n_{\succ 0}$. Vectors $p$ of length $n$ are $p \in \mathbb{R}^n$.  The space of diagonal matrices of dimension $n$ are $\mathbb{D}^n_+$. The identity matrix of dimension $n$ is $I_n$.  
 Symmetric matrices $A$ of dimension $n$ are denoted $A \in \mathbb{S}^n$. 
 The $i$-th row of a matrix $A$ is denoted $A_{(i,:)}$. The $i$-th element of a vector $x \in \mathbb{R}^n$ is likewise denoted $x_i$. The set of natural numbers with zero is denoted $\mathbb{N}_0 = \mathcal{N} \cup \{0\}$. Given a nonlinear function $\phi\colon\Rset^n\to\Rset^n$, the notation $\phi(s)\in\Rset^{n}$ for $s = [s_1,\ldots,s_n]^\top\in\Rset^{n}$ refers to a component-wise application, \ie $\phi(s) = [\phi(s_1),\ldots,\phi(s_n)]^\top$.

\section{Linear Quadratic MPC }

Consider the discrete linear time-invariant state-space model
\begin{align}\label{sys}
x[k+1]  = Ax[k]+ Bu[k],
\end{align}
where $x[k]\in\Rset^n$ and $u[k]\in\Rset^m$ are the state and control respectively evolving on time instants $k~=~0,\,1,\,2,\,\dots$. We use the  notation $x[k] = x$ throughout to simplify the analysis. A state ${x}\in\Rset^{n}$ and an $N-$step sequence of control actions $\mb{u} = [u[k]^\top,\ldots,u[k+N-1]^\top]^\top$ generate the $N-$step state predictions sequences $\mb{x} = [{x}^\top,x[k+1]^\top,\ldots,x[k+N]^\top]^\top$. Here, the inputs $\mb{u}$ are 
 constrained by the state-dependent polytopic set $\mc{U}(x)= \{\mathbf{u}: G\mathbf{u}  \leq S_ux+w\}$ for some $G \in \mathbb{R}^{n_u \times N m}$, $S_u \in \mathbb{R}^{n_u \times n}$ and $w \in \mathbb{R}^{n_u}$.

\subsection{MPC problem}

From a given state $x\in\Rset^n$, the MPC control input sequence $\mb{u}(x)$ is computed from the quadratic programme (QP):
\begin{subequations}
  \begin{align}
    \min_{\mb{u}(x)}~& x[k+N]^{\top}Px[k+N]+\sum_{i=1}^{N-1} x[k+i]^{\top}  Q_ix[k+i] + u[k+i]^{\top} R_iu[k+i] 
    \\   \ts{subject to}~  x[k] &= x, \nonumber
    \\ G\mb{u}  &\leq S_ux+ w,
    \\  x[i+1] &= Ax[i] + \tilde{B}_i\mb{u},
  \end{align}
  \label{QP1}
\end{subequations}
where $\tilde{B}_i = [A^{i-1}B~~\ldots~ AB~~B~~\mb{0}_{n\times (N-i-1)m}]\in\Rset^{n\times mN}$,  $Q_i\in \Rset^{n\times n}$, $P\in \Rset^{n\times n}$, $R_i\in \mathbb{R}^{m\times m}$, $Q_i \succeq 0$,  $R_i \succ 0$, $\forall i  = 1,\ldots,  N -1$, and $P\succ 0$. Exploiting the linear dynamics of \eqref{sys} allows the above to be equivalently written as the following QP (referred to here as the \textit{MPC-QP}).
\begin{definition}[MPC-QP]\label{def:mpc-qp} With a positive-definite Hessian $H \in \mathbb{S}_{\succ 0}^{N}$, the solution of the MPC problem at a state $x\in\Rset^{n}$ defined in \eqref{QP1} is given by
\begin{subequations}\label{eq:mpcQP}\begin{align}
    \mathbf{u}^*(x) =  & \arg \min_{ \mathbf{u}} \mathbf{u}^{\top} H \mathbf{u} + 2x^{\top}F^{\top}\mathbf{u}, \\
     & \text{subject to } G\mathbf{u}  \leq S_ux+w.
\end{align}\end{subequations}
\end{definition}
The solution of \textit{MPC-QP} for a given state $x\in\Rset^n$ is referred to here as \textit{implicit MPC}, since the computed optimal input sequence is an implicit function of the current value of the state $x$, \ie $\mathbf{u}^*(x)$. To obtain MPC control action $u_\ts{mpc}\in\Rset^m$ from the vector $\mathbf{u}^*(x)$, we use
\begin{align}
    u_\ts{mpc} = [1,\,0,\,\dots, 0]\mathbf{u}^*(x).
\end{align}


{\flushleft \textbf{Assumption 1}}
  {It is assumed that a unique solutions $  \mathbf{u}^*(x) $ exists for the MPC-QP problem of \eqref{eq:mpcQP} for all $x \in \mathbb{R}^n$. This feasibility assumption will translate into an assumption on the well-posedness of the equivalent implicit neural network derived in this paper. 
  }


\section{Implicit representation of MPC}\label{sec:MPC}
 
%

As discussed in the introduction,  there has been a concentrated effort to link the disparate control schemes of MPC and NN controllers,  so as to accelerate, scale and ease the implementation of MPC. As far as the authors are aware, existing studies relating NNs and MPC, such as \cite{fahandezh2020proximity,xu2021lattice,karg2020efficient,ferlez2020aren,aakesson2006neural,kittisupakorn2009neural,hertneck2018learning,soloperto2020augmenting}, are motivated by the fact that the explicit MPC policy is a piece-wise affine function, so should be well-approximated by an \textit{explicit NN}, in particular one that is also a piece-wise affine function with \emph{ReLU} activation functions.

\begin{definition}\label{def:exp}
  An explicit NN $f\colon\Rset^{n}\to\Rset^{Nm}$ with $x\mapsto\mb{u}^*(x)$ is defined by the recurrence
\begin{subequations}\begin{align}
    z[j+1] &= \phi\left(\sum_{i = 0}^{j}W^{j,i}z[i]+Y^{j}x+b^j\right),\label{eq:zexplicit}
    \\
    \mathbf{u}^*(x) &= W_fz[J] + Y_fx + b_f,
\end{align}\end{subequations}
for some activation function $\phi(\cdot)$, hidden layer depth $j = 0,\, \dots,\, J$, $W^{j,i} \in \mathbb{R}^{n_j \times n_i}$ (where $n_i$ and $n_j$ are the dimensions of layers $i$ and $j$, respectively),  $Y^{j} \in \mathbb{R}^{n_j \times n}$,  $b^j \in \mathbb{R}^{n_j}$, $W_f \in \mathbb{R}^{Nm \times n_J}$, $Y_f\in \mathbb{R}^{Nm \times n}$  and $b_f \in \Rset^{Nm}$. 
\end{definition}
Several methods have been proposed to train explicit neural networks that approximate MPC, e.g. in~\cite{zhang2020near,chen2022large,karg2021approximate}.
But, in general, these methods all require the standard training step to obtain the NN from data generated by sampling the MPC solution. Training in this way can introduce errors which are challenging to bound \cite{drummond2022bounding, fabiani2022reliably,schwan2023stability}.

By contrast, the focus of this paper is on \textit{equating} MPC policies to the \textit{implicit} neural networks of~\cite{el2021implicit}. The following definition of an implicit neural network is used here. 


\begin{definition}\label{def:imp}
An implicit NN $f\colon\Rset^{n}\to\Rset^{Nm}$ with $x\mapsto\mb{u}^*(x)$ is defined by the solution of the following algebraic system of equations
\begin{subequations}\label{imp_def_2}\begin{align}
    y(x) &= W \phi(y(x))+Yx+ b,\label{imp_top}\\
    \mathbf{u}^*(x) &= W_f\phi(y(x)) +Y_{f}x,
\end{align}\end{subequations}
for some activation function $\phi(\cdot)$, weights $W\in \mathbb{R}^{M \times M}$, $Y\in \mathbb{R}^{M \times n}$, $W_f\in \mathbb{R}^{Nm \times M}$, $Y_f\in \mathbb{R}^{Nm \times n}$ and biases   $b\in \mathbb{R}^{M}$,  $b_f\in \mathbb{R}^{Nm}$. 
\end{definition}
%
The \emph{implicit} term comes from the fact that~\eqref{imp_top} is an implicit equation in general (in certain cases, such as when the weight matrix  $W\in \mathbb{R}^{M \times M}$ has a strict upper- or lower-triangular structure, then $\mathbf{u}^*(x)$ again becomes an explicit function of $x$, as discussed in \cite{el2021implicit}). Compared to explicit NNs, such as the classical feed-forward and recurrent architectures, evaluating an implicit NN may involve solving a system of nonlinear equations, since the weight matrices $W\in \mathbb{R}^{M \times M}$ may be dense. 
Implicit equations can be difficult to solve and conditions for the well-posedness of the implicit equation must be verified. Sufficient conditions to establish well-posedness of the implicit neural network of \eqref{imp_def_2} can be derived by adapting the method from  \cite{el2021implicit} involving the weight matrix $W$. 




\begin{remark}
A key aspect of this paper is the observation that, in general, the solution of an optimal control problem is the output of an \textit{implicitly} defined function.  Only in particular cases, such as the robust control problems of  \cite{pates2019optimal}, are explicit solutions to control problems known. For the other cases, numerical algorithms are required to extract the solutions. The implicit nature of solving optimal control problems suggests that implicit neural networks should provide a more natural framework to express the solutions than the more commonly used explicit neural networks. 
\end{remark}


The main result of this paper is to show that MPC problems defined by the solution of the \textit{MPC-QP} admit an exact interpretation in terms of an  \textit{implicit neural network} of the form of Definition \ref{def:imp}. The rest of the paper is concerned with showing this.  The main idea is to build on top of the recent result from \cite{valmorbida2023quadratic} where it was shown that the \textit{MPC-QP} can be expressed as a linear complimentarity problem using \emph{ReLU} functions.

%
\begin{definition}
A function $r(s):  \mathbb{R}^{n_s} \to \mathbb{R}^{n_s}$ is a \emph{ReLU} function if it satisfies
\begin{equation}\label{eq:ramp}
r(s) = \left\{ \begin{array}{l} 0 \quad \mathrm{if} \, s < 0, \\ s \quad \mathrm{if} \, s \geq 0. \end{array} \right. 
\end{equation}
This function is also known as the ramp function. 
\end{definition}
%
\begin{theorem}{\cite{valmorbida2023quadratic}} \label{thm:MPCramp}
The solution of the \textit{MPC-QP}~\eqref{eq:mpcQP} at a state $x\in\Rset^{n}$ can be expressed as the solution of the piece-wise affine system of equations
\begin{subequations} \label{eq:uPWAramp}
\begin{align}
\mathbf{u}^*(x)& =- H^{-1}F^{\top} x  - H^{-1}G^{\top} r(y(x)), \label{eq:uPWArampU}
\end{align}
where $y(x): \mathbb{R}^n \rightarrow \mathbb{R}^{n_c}$ is the solution to the following implicit equation
\begin{align}
y(x) - (I - GH^{-1}G^{\top} )r(y(x))& =   -(Sx  +w),\label{unravel1}
\end{align}
\end{subequations}
with $S = S_u+GH^{-1}F$. 
\end{theorem}
In the above theorem,  the state $x$ and the optimal MPC control input $\mathbf{u}^*(x)$ are connected through the implicitly defined piece-wise affine function $y(x)$ which solves \eqref{unravel1}. 


\begin{remark}
 The expression~\eqref{unravel1} is obtained from the MPC-QP Karush-Kuhn-Tucker optimality conditions
\begin{subequations}\label{KKT-MPC}
\begin{align}
\label{eq:z}
H\mathbf{z} +G^\top\lambda &= 0,
\\
\label{eq:complKKT}
\lambda_i(S_{(i, \cdot)}x +w_i -G_{(i, \cdot)} \mathbf{z}) &= 0,\\
\label{eq:constposKKT}
Sx +w -G\mathbf{z} &\geq 0,
\end{align}
\end{subequations}
which, thanks to the invertibility of $H$, gives the following  Linear Complementarity Problem
\begin{align}
\label{eq:LCP}
0 \leq \lambda \perp (Sx +w +GH^{-1}G^\top\lambda) &\geq 0.
\end{align}
Relating the Lagrange multipliers $\lambda$ of the \textit{MPC-QP} to the ramp function of some functions of $x$, namely $\lambda = r(y(x))$, then the  complementarity conditions stated above can be connected to the complementarity conditions defining the ramp function, giving~\eqref{unravel1}. See~\cite{valmorbida2023quadratic} for details.
\end{remark}


Equation~\eqref{unravel1} allows the MPC-QP feedback law to be characterised as an implicit function.  An iterative method to solve~\eqref{unravel1} is proposed in~\cite{valmorbida2023quadratic}, providing an alternative to using interior-point methods to solve the point-location problem of explicit MPC~\cite{bemporad2002explicit,Bayat2011}. In the rest of the paper, we exploit the ReLU function in~\eqref{unravel1} to propose an approximate solution to the implicit equation~\eqref{unravel1}. 


By equating terms between Equation \eqref{imp_def_2} of Definition \ref{def:imp} and Equation \eqref{eq:uPWAramp} of Theorem \ref{thm:MPCramp}, it follows that the MPC control action defined by the solution of the \textit{MPC-QP} (Definition \ref{def:mpc-qp}) can be uniquely represented as an implicit neural network  of the form of Definition~\ref{def:imp} with 
\begin{subequations}\label{lem:imp_NN}\begin{align}
    W &= (I_m - GH^{-1}G^{\top} ),  \,Y = -S,\, b = -w , \\
    W_f &= -H^{-1}G^\top,\, Y_f = - H^{-1}F^{\top}, \, b_f = 0,
\end{align}\end{subequations}
and with a \emph{ReLU} activation function $\phi(\cdot) = r(\cdot)$. In fact, from this point on, the activation functions considered in the paper will be \emph{ReLU}s. The matrices in~\eqref{lem:imp_NN}  show that to link MPC and NNs, an exact ``\textit{implicit-to-implicit}'' mapping between MPC and NNs can be obtained, unlike the  ``\textit{explicit-to-explicit}'' mappings usually considered which either involve approximation errors, \eg in  \cite{fahandezh2020proximity,karg2020efficient}, or are restrictive, e.g. for single state systems \cite{moritz2}. This observation connecting the implicit neural network weights \eqref{lem:imp_NN} to the \textit{MPC-QP} is the main result of this paper.

\section{Explicit approximation of the implicit neural network} \label{sec:explicit}

While the implicit NN of Definition \ref{def:imp} parameterised by \eqref{lem:imp_NN} exactly captures the solution of the \textit{MPC-QP},  it can be argued that implicit NNs are currently not yet widely used in practice. This is primarily because evaluating them requires solving an implicit system of equations which can be challenging. Another issue with implicit neural networks is that most  of the common NN architectures used for control, such as recurrent and feed-forward neural networks, are explicit and have many well-established tools for training, evaluating and pruning them, unlike for implicit NNs. 
These limitations motivate the following results  on ``\textit{unravelling}'' the implicit NN of Equation \eqref{lem:imp_NN} into an approximating explicit NN. 


The first step towards unravelling the implicit NN of Definition \ref{def:imp} with \eqref{lem:imp_NN} is to write \eqref{unravel1} as 
\begin{align}\label{exp_unrav}
    y(x) = D\phi(y(x)) + \zeta
\end{align}
with $D \in \mathbb{S}^{m} = (I_m-GH^{-1}G^\top)$ and $\zeta \in \mathbb{R}^m= -(Sx+w)$.  The interpretation of the implicit NN as the equilibrium  of a forced dynamical system, \cite{el2021implicit, revay2023recurrent}, is adopted to generate an approximating explicit NN. With this interpretation, the approximating explicit NN is structured as
\begin{align}\label{explicit:iterates_orig}
    w[j+1] = D\phi(w[j]) + \zeta + f(w[j]),
\end{align}
with the hidden layer index $j = 0, \,1,\, \dots,\, J$. In the above, the function  $f(\cdot)$  is a feedback control policy chosen by the user to accelerate the convergence  of $w[j]$ towards $y(x)$. The benefits of accelerated convergence are that shallower neural networks can be used for a given level of solution accuracy. 

\subsection{Error dynamics}
The problem is then to design the control action $f(\cdot)$ such that the dynamics of the approximating explicit NN \eqref{explicit:iterates_orig} converge to the output of the implicit one \eqref{exp_unrav}. For this,  the solution of the implicit NN \eqref{exp_unrav} is subtracted from the iterates of \eqref{explicit:iterates_orig} to give the error $e[j] = w[j] -y(x)$ evolving through the layers $j$ by $e[j+1] =   D(\phi(w[j])-\phi(y(x))+ f(w[j])$.
In this paper, the controller is structured as $f(w[j]) = K(w[j]-D\phi(w[j])-\zeta)$ for some gain $K\in \mathbb{R}^{M\times M}$. If the error dynamics are convergent with this controller, then the implicit neural network can be understood as the output of the infinitely long explicit neural network. Truncating this infinitely long explicit neural network after a certain number of time steps leads to a  feed-forward neural network with the explicit structure of Definition \ref{def:exp} and a bounded approximation error.






To show this, define
\begin{align}\label{eqn:phi_tilde}
    \tilde{\phi}(e[j]) = \phi(w[j])-\phi(y(x)),
\end{align}
such that the iterates evolve according to
\begin{align}\label{explicit:iterates}
    w[j+1] = D\phi(w[j]) + \zeta + K(w[j]-D\phi(w[j])-\zeta),
\end{align}
and the error dynamics satisfy
\begin{align}
    e[j+1] = &  D(\phi(w[j])-\phi(y(x))+  K(w[j]-D\phi(w[j])-\zeta), \nonumber
    \\
     = & Ke[j] + (I_m-K)D\tilde{\phi}(e[j]). \label{eqn:error_dyns_K}
\end{align}
This is a nonlinear system but notice that $\tilde{\phi}(e[j])$ is a slope-restricted nonlinearity.
\begin{lemma}\label{def:sector}
The function $\tilde{\phi}(e[j]): \mathbb{R}^{n_s}\to \mathbb{R}^{n_s}$ defined in \eqref{eqn:phi_tilde}
 satisfies $\tilde{\phi}(e[j])^\top T(e[j]-\tilde{\phi}(e[j]))~\geq~0$ for all $T \in \mathbb{D}^{m}_+$.
\end{lemma}
\begin{proof}
First, note that the \emph{ReLU} function satisfies the equality
\begin{align}\label{ReLU:eq}
    \phi(s)^\top T(s-\phi(s)) = 0, \quad \forall s \in \mathbb{R}^m, T \in \mathbb{D}_+^m.
\end{align}
Then
\begin{subequations}\begin{align}
    \tilde{\phi}(e[j])^\top T(e[j]-\tilde{\phi}(e[j])) & 
    =-(\phi(w[j]))^\top T(y(x)-\phi(y(x))) ,
    \\ & =  -\phi(y(x)))^\top T(w[j]-\phi(w[j])) \geq 0,\nonumber
\end{align}\end{subequations}
since $T \in \mathbb{D}^m_+$, the \emph{ReLU} function satisfies $\phi(s) \geq 0$, $s-\phi(s) \leq 0$ for all $s \in \mathbb{R}$  and using \eqref{ReLU:eq}.
\end{proof}


The error dynamics of \eqref{eqn:error_dyns_K} can then be understood as a discrete-time Lurie system. As such, for a given gain $K$, the linear matrix inequalities from papers such as \cite{haddad1993explicit,carrasco2019convex,drummond2023generalised} could then be used to certify asymptotic convergence and the errors bounded. If these LMIs are satisfied, then the feed-forward neural network of \eqref{explicit:iterates} is guaranteed to converge towards the solution of the \textit{MPC-QP} of Definition \ref{def:mpc-qp}. Page limits of this paper prevent an in-depth analysis of the computation of the gains $K$ which is left to future work. 

\begin{remark}
Several results, such as \cite{karg2020efficient,cao2020deep}, have emphasized the importance of explicit neural network \textit{depth}, as in the number of hidden layers, to gain a good approximation of an MPC control law. The analysis discussed in this section reinforces this notion, by showing how an infinitely deep feed-forward neural network, of the form of \eqref{explicit:iterates}, can exactly represent the solution of the \textit{MPC-QP}. 


\end{remark}

\section{Converse results: From neural networks to MPC}
\label{sec:conv-results:-from}

In this section, a converse result to the statement of Section~\ref{sec:MPC} is introduced. Specifically,  the following problem is considered: Given an implicit neural network $f\colon\Rset^n\to\Rset^{Nm}$ of the form of Definition~\ref{def:imp}, can the cost defining the MPC control law be recovered? The answer to this question is affirmative. The argument is the following: the piece-wise affine functions are dense in the set of continuous functions, therefore there exists a piece-wise linear approximation $\xi\colon\Rset^{n}\to\Rset^{Nm}$ to $f(\cdot)$ such that $\|f - \xi \|< \epsilon $ in a suitable norm. The function $\xi(\cdot)$ generates a partition of the state space $\tilde{P}_i$ for $i\in\tilde{\mc{I}}$ and a collection of affine control laws $\{\tilde{E_i},\tilde{\omega}_i\}_{i\in\tilde{{P}}}$. To recover the cost matrices, an observation regarding the partition of the state space is made.
\begin{lemma}
  There exists a non-empty proper subset $\emptyset\neq\mc{J}\subset\mc{I}$ such that the associated control laws satisfy  $j\in\mc{J}$, $\kappa_j(x) = \omega_j$ for each $x\in\tilde{P}_j$.
  \label{lem:control_max}
\end{lemma}
The importance of the above lemma is that focus only has to be given to the interior of the feasible region, with the outer parts representing a ``saturation'' of the control input. In the interior of $\mc{X}$, there exists a subset of indices $\mc{K} = \mc{I}\setminus\mc{J}$ such that $\kappa_i(x) = K_ix + \beta_i$ for all $i\in\mc{K}$. The closed loop system in each of the regions $P_i$ is characterised by $\Phi_i = (A+BE_i)$ which satisfy the following LMI:
\begin{equation}
  \Phi_i^\top P_i\Phi_i - P_i \succ - (Q + E_i^\top R E_i),~\forall i\in\mc{K},
  \label{eq:LMI_partition}
\end{equation}
with $P_i,Q,R\prec 0$ for all $i\in\mc{K}$. The solution of this system of LMIs allow us to build common cost matrices $(Q,R)$ that are part of the MPC formulation.


\section{Numerical Example}
Consider the linear time-invariant system from \cite{drummond2022bounding}
\begin{equation}
    \begin{bmatrix}x_1[k+1] \\ x_2[k+1] \end{bmatrix}=
        \begin{bmatrix}
            4/3 & -2/3\\
            1 & 0
        \end{bmatrix}
      \begin{bmatrix}x_1[k] \\ x_2[k] \end{bmatrix}+
        \begin{bmatrix}
            0 \\ 1
        \end{bmatrix}
    u[k],
\end{equation}
to be controlled using an MPC defined by the \textit{MPC-QP} of Definition \ref{def:mpc-qp}.
%
%
%
Set the horizon length $N = 10$ and the control action to be saturated at $-10\leq \mathbf{u} \leq 10$, so $G~=~I_m~\otimes~\begin{bmatrix}0.1 & -0.1 \end{bmatrix}^\top$, $w=[1,\,  1,\, \dots \,1 ,\, 1 ]^\top$ and $S_u = 0$. The MPC quadratic cost function of \eqref{QP1} is parameterised by
$$
\widetilde{P} =  \begin{bmatrix}7.1667  & -4.2222\\ -4.2222 & 4.6852\end{bmatrix},~    \widetilde{Q}_k = \begin{bmatrix}1  & -2/3\\ -2/3  & 3/2\end{bmatrix},~
\widetilde{R}_k = 1.
$$
The goal of the numerical example is to verify the claim that this MPC problem can be represented as the implicit neural network of Definition \ref{def:imp} defined by \eqref{lem:imp_NN} which can itself be approximated by the explicit one of \eqref{explicit:iterates}.

 \begin{figure}[t]
         \centering
         \includegraphics[width=1\textwidth]{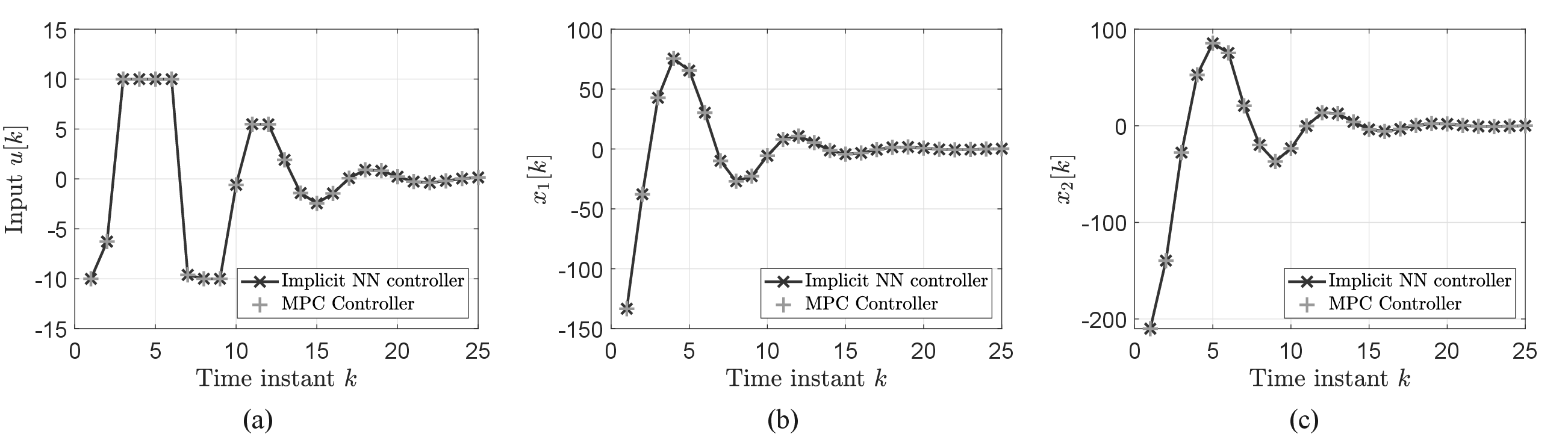}
          \caption{Comparison between the MPC controller and the implicit neural network defined by \eqref{lem:imp_NN} for the numerical example. The input sequences $u[k]$, and states $x_1[k]$, $x_2[k]$ are shown for both policies.  The equivalence between the two simulations validates the results of the paper. Note that the implicit neural network captures the input constraints imposed by the MPC policy, as predicted. }
          \label{fig:sim}
     \end{figure}


The first step is to show that the explicit neural network of \eqref{explicit:iterates} does indeed  converge towards the unique solution of the implicit neural network of \eqref{exp_unrav}. From the initial condition $w[0] = [1,1]^\top$, the curve in Figure \ref{fig:combo}(a) shows the convergence rates of the various methods defined by the ``control gains'' $K$. Three different gains are considered, $K = -0.9I_m$, $K = 0$ and $K = 0.2I_m$. The figure shows the convergence of the residuals $w[j]-D\phi(w[j])-\zeta$ for each gain, indicating that, for a sufficient layer depth $J$, the explicit neural networks did converge to the unique solutions of the implicit neural network. Fastest convergence was achieved with $K = -0.9I_m$. LMI-based conditions to optimise the gains $K$ will be developed in future work.  


With the explicit neural networks of  \eqref{explicit:iterates} shown to give a good approximation of the implicit neural networks of \eqref{exp_unrav} (as long as they are sufficiently deep), the next step is to show that this implicit neural network (parameterised according to \eqref{lem:imp_NN}) represents the MPC control action defined by the \textit{MPC-QP} of Definition \ref{def:mpc-qp}. This is illustrated in the simulation of Figure \ref{fig:sim} as well as in Figures \ref{fig:combo}(b) and  \ref{fig:combo}(c) which also show the equivalence, with the two control actions compared for various different state values. For this simulation, the initial condition was  $x = [-200,\,-200]^\top$, and controllers defined by the MPC of Definition \ref{def:mpc-qp}  and the explicit neural network of \eqref{explicit:iterates} were run. Using the depth $J = 1000$,  Figure \ref{fig:sim}, \ref{fig:combo}(b) and  \ref{fig:combo}(c) shows the equivalence between the inputs $u[k]$ generated by both the MPC and NN controllers (with both saturating at certain instants $k$), which then also generated equivalent state dynamics for both $x_1[k]$ and $x_2[k]$. The equivalence between these two closed-loop trajectories reinforces the conclusion of Equation \eqref{lem:imp_NN} that the \textit{MPC-QP} can be represented as the implicit neural network of \eqref{exp_unrav}, and that the analytic expression of Equation \eqref{lem:imp_NN} for the weights and biased hold.

 \begin{figure}[t]
         \centering
         \includegraphics[width=1\textwidth]{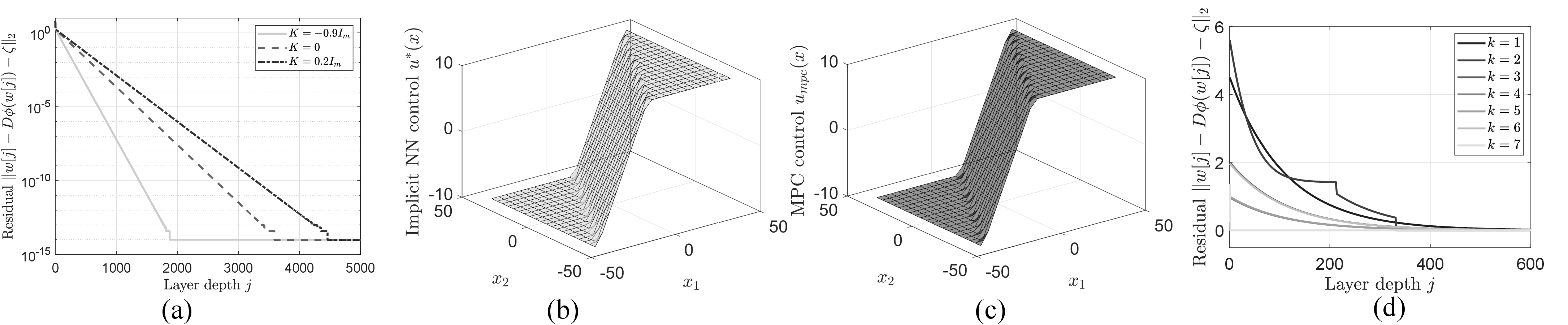}
          \caption{(a) Residuals from unravelling the implicit neural network into an explicit one using the gains $K = -0.9I_m$, $K = -0.9I_m$ \& $K = 0.2I_m$. (b) The control action from the implicit neural network of the numerical example as a function of the state-space. (b) The control action from the MPC of the numerical example as a function of the state-space. (d) Residuals of the implicit neural network during the simulation as $k$ evolves. }
       \label{fig:combo}
     \end{figure}

 Finally, Figure \ref{fig:combo}(d) shows the convergence of the approximating explicit neural network's residuals during the simulation (as in, as $k$ increased). In the simulation, the explicit neural network was warm started, meaning that the initial condition $w[0]$ was based upon the final value $w[J]$ obtained from the previous time instant $k$. Figure \ref{fig:combo}(d) shows the accelerated convergence of these residuals as the simulation progressed with $k$ increasing, and these residuals were found to be zero for $k >7$. This figure shows the accelerated convergence rates of the explicit neural network as the state dynamics converge towards their equilibrium.

\section*{Conclusions}
Linear quadratic model predictive control (MPC) was shown to admit a representation as an implicit neural network. Analytic expressions for the weights and biases of this implicit neural network were given and a method to ``\textit{unravel}'' the implicit neural network into an explicit one was also described. Expressions for the weights and biases of the approximating explicit neural network were stated and it was shown that the depth of the explicit network defined the approximation error with respect to the implicit neural network. A numerical example  illustrated the application of the results in practice. A procedure to recover the MPC cost matrices from piece-wise linear approximations of explicit neural networks was also developed. Together, these results give a constructive way to represent linear quadratic model predictive control using neural networks, providing an explicit link between model-based and data-driven control.

\acks{Ross Drummond was supported by a UK Intelligence Community Research Fellowship from the Royal Academy of Engineering. }

\bibliography{sample.bib}

\end{document}